\documentclass[11pt,a4paper]{article}

\usepackage[utf8]{inputenc}
\usepackage{amsmath}
\usepackage{amsfonts}
\usepackage{amssymb}
\usepackage{amsthm}
\usepackage{algorithm}
\usepackage{algpseudocode}
\usepackage{multicol}

\title{Reachability Problems for Continuous Chemical Reaction Networks\footnote{This research was supported in part by National Science Foundation Grant 1247051. Part of the second author’s work was carried out while participating in the 2015 Focus Semester on Computability and Randomness at Heidelberg University.}}
\date{}
\author{Adam Case, Jack H. Lutz, and D. M. Stull\\
		\small{Department of Computer Science}\\
		\small{Ames, IA 50011 USA}\\
		\small{$\{$adamcase, lutz, dstull$\}$@iastate.edu}}

\newtheorem*{definition}{Definition}
\newtheorem*{ctsReachdefinition}{The Continuous CRN Reachability Problem}
\newtheorem*{subsetReachdefinition}{The Continuous CRN Subset Reachability Problem}

\newtheorem{observation}{Observation}

\newtheorem{lemma}{Lemma}

\newtheorem{maintheorem}{Main Theorem}
\newtheorem{reftheorem}{Theorem}

\theoremstyle{definition}
\newtheorem{prooflemma1}{Proof of Lemma}

\begin{document}

\maketitle

\begin{abstract}
Chemical reaction networks (CRNs) model the behavior of molecules in a well-mixed system. The emerging field of molecular programming uses CRNs not only as a descriptive tool, but as a programming language for chemical computation. Recently, Chen, Doty and Soloveichik introduced a new model of chemical kinetics, rate-independent continuous CRNs (CCRNs), to study the chemical computation of continuous functions. A fundamental question of a CRN is whether a state of the system is reachable through a sequence of reactions in the network. This is known as the reachability problem. In this paper, we investigate CCRN-REACH, the reachability problem for this model of chemical reaction networks. We show that, for continuous CRNs, constructing a path to a state of the network is computable in polynomial time. We also prove that a related problem, Sub-CCRN-REACH, is NP-complete.

\end{abstract}

\section{Introduction}

Abstract chemical reaction networks (CRNs) model chemical interactions in a well mixed system. Informally, CRNs consist of a finite set of species of chemicals (usually written abstractly as capital letters; i.e., A, B, etc.) and a finite set of reactions between these species. A simple example is the CRN consisting of species $A$, $B$ and $C$, with one reaction $2A + B \xrightarrow{k} 2C$ (taking $A$ to be the hydrogen molecule $H_2$, $B$ to be the oxygen molecule $O_2$, and $C$ to be the water molecule $H_2O$, this CRN models the formation of water molecules with kinetic rate constant $k$). CRNs have historically been used as a descriptive tool, allowing researchers to formally analyze the behavior of natural chemical systems. However, the field of molecular programming has recently brought CRNs to prominence as a programming language for chemical computation. Molecular programming, as the name suggests, is devoted to engineering complex computational systems from molecules. Recent work in this area has come to view abstract CRNs as a programming language to engineer ``chemical software" \cite{JRP, SCWB}. Exciting new developments have shown methods of converting arbitrary chemical reaction networks into computation using DNA strands \cite{Cardelli, CDSPC, SSW}. Thus the programmable power of chemical reaction networks is no longer simply of theoretical interest. To achieve the goal of engineering large scale, robust chemical computation, tools to analyze CRNs will be vital. 

There are many ways to define the behavior of abstract CRNs, the two most prominent being mass action kinetics and stochastic chemical reaction networks. Mass action kinetics was the first model to be studied extensively. It is a continuous, deterministic model of chemical reaction networks. Mass action kinetics is used to study systems with sufficiently large numbers of molecules so that the amount of a given molecule can be represented as a real-valued concentration. The dynamics of reactions under mass action kinetics are governed by deterministic differential equations. However, the deterministic mass action model is not well suited if the number of molecules of the system is low. Stochastic CRNs are widely used to analyze those systems with a relatively low number of molecules \cite{MA, ELSS}. The stochastic CRN model is discrete and non-deterministic. Unlike mass action, the amount of each species is represented as a non-negative integer, and the reactions of a system are modeled as Markov jump processes \cite{gillespie}. The stochastic model is closely related to many well-studied models of computation such as Vector Addition Machines \cite{KM}, Petri Nets \cite{EN} and Population Protocols \cite{AADFP}. Recently, Chen, Doty and Soloveichik introduced a new model of chemical kinetics, rate-independent continuous CRNs (CCRNs) \cite{CDS}. The CCRN model is continuous, dealing with real-valued concentrations of species, but, unlike the stochastic or mass action models, it is rate-free (reactions do not have any associated kinetic rate constant). Chen, Doty and Soloveichik used CCRNs to study which real valued functions $f: \mathbb{R}^k \rightarrow \mathbb{R}$ were computable by a chemical reaction network. By being a rate-free model, it allows for the study of the computational power of large chemical systems relying on stoichiometry alone (i.e., without depending on specific rates of the reactions). This is important, as rate constants are hard to experimentally determine and vary under external factors such as temperature.

A fundamental question one can ask of a stochastic chemical reaction network is whether a particular state is reachable from a starting configuration; this is called the reachability problem. The reachability problem for stochastic CRNs is equivalent to an important problem in theoretical computer science, the Vector Addition System Reachability problem (VAS reachability) \cite{CSWB}. The VAS reachability problem was proven to be at least EXPSPACE-hard by Lipton in 1976 \cite{lipton}, before it was even proven decidable. In 1981, building on the work of Sacerdote and Tenney \cite{ST}, Mayr proved the reachability problem was decidable \cite{mayr}. Subsequently, Kosaraju \cite{kosaraju} and Lambert \cite{lambert} gave two additional proofs of the decidability of VAS Reachability. However, all proofs that the reachability problem is decidable were very difficult, until L\'er\^oux \cite{leroux} gave a greatly simplified proof. Unfortunately, we still do not know if this problem is decidable in any primitive recursive time bound. 

In this paper, we investigate two variants of the reachability problem in the context of CCRNs. In section 3, we analyze the complexity of the direct analog of the reachability problem for CCRNs, the continuous chemical reaction network reachability problem (CCRN-REACH). Informally, the CCRN-REACH problem is: given a CCRN $C$ and states $\mathbf{c}$ and $\mathbf{d}$, output a path taking $\mathbf{c}$ to $\mathbf{d}$, if one exists. To effectively compute CCRN-REACH, we will require the states to be over the rationals instead of over arbitrary reals. We show that, contrary to the difficulty of the VAS reachability problem, CCRN-REACH can be computed in polynomial time. In the process, we give new definitions and lemmas which we believe will be useful in further investigations of the continuous chemical reaction network model. In section 4, we define a problem closely related to the reachability problem, called the Sub-CCRN-REACH problem. Sub-CCRN-REACH asks if a path exists between two states using at most $k$ of the reactions in the network. In contrast to the computational ``ease" of CCRN-REACH, we show that Sub-CCRN-REACH is NP-complete.

\section{Preliminaries}
Throughout the remainder of this paper $\|\cdot\|$ will be the max norm. Before proving the main theorem, we will review preliminary definitions and notations for continuous CRNs.
\subsection{Rate Independent Continuous CRNs}
A \textit{continuous chemical reaction network (CCRN)} is a pair $C = (\Lambda, R)$, where $\Lambda$ is a finite set of \textit{species}, and $R$ is a finite set of \textit{reactions} over $\Lambda$. We typically denote species by capital letters, so that $\Lambda = \{A, B,\ldots\}$. A \textit{reaction} over the set of species $\Lambda$ is an element $\rho = (\mathbf{r}, \mathbf{p}) \in \mathbb{N}^{\Lambda} \times \mathbb{N}^{\Lambda}$, where $\mathbf{r}$ and $\mathbf{p}$ specify the stoichiometry of the reactants and products, respectively. We require the \textit{net change} $\Delta\rho = \mathbf{p} - \mathbf{r}$ of a reaction $\rho = (\mathbf{r}, \mathbf{p})$ to be nonzero. We will usually write a reaction using the ``reactants, right arrow, products" notation; for example, $\rho = A + B \rightarrow C$ (in this example $\mathbf{r} = (1, 1, 0)$ and $\mathbf{p} = (0, 0, 1)$). A reaction $\rho = (\mathbf{r}, \mathbf{p})$ is \textit{catalytic} if, for some species $s$, $\mathbf{r}(s) = \mathbf{p}(s) \neq 0$ (for example, $A + B \rightarrow A + C$). In this case, we call the species $s$ a \textit{catalyst}. Each CCRN $C = (\Lambda, R)$ has an associated \textit{reaction stoichiometry matrix} $\mathbf{M}$ specifying the net change of each species for every reaction. Formally, $\mathbf{M}$ is a $\vert \Lambda \vert$ $\times$ $\vert R \vert$ matrix over $\mathbb{Z}$ such that $\mathbf{M}(i, j)$ is the net change of the $i$th species for the $j$th reaction. Note that $\mathbf{M}$ does not fully specify a CCRN $C$, since it does not identify catalytic reactions. A \textit{state} of a CCRN $C = (\Lambda, R)$ is a vector $\mathbf{c} \in \mathbb{R}_{\geq 0}^{\Lambda}$ specifying the (non-negative) concentration of each species. The \textit{support of a state} $\mathbf{c}$ is the set $supp(\mathbf{c}) = \{s \in \Lambda \, \vert \, \mathbf{c}(s) > 0\}$ of all species with non-zero concentrations at $\mathbf{c}$. The \textit{support of a reaction} $\rho = (\mathbf{r}, \mathbf{p})$ is the set $supp(\rho) = \{s \in \Lambda \, \vert \, \mathbf{r}(s) > 0\}$ of all reactants of $\rho$. A reaction $\rho = (\mathbf{r}, \mathbf{p}) \in R$ is \textit{applicable} at a state $\mathbf{c}$ if $supp(\mathbf{r}) \subseteq supp(\mathbf{c})$ (i.e., if the concentration of each reactant is non-zero at $\mathbf{c}$). A \textit{flux vector} of a CCRN $C = (\Lambda, R)$ is a vector $\mathbf{u} \in \mathbb{R}_{\geq 0}^{R}$. The \textit{support of a flux vector} $\mathbf{u}$ is the set $supp(\mathbf{u}) = \{\rho \in R \,|\, \mathbf{u}(\rho) > 0\}$. A flux vector $\mathbf{u}$ is \textit{applicable at a state $\mathbf{c}$} if the following conditions hold:
\begin{enumerate}
\item Every $\rho \in supp(\mathbf{u})$ is applicable at $\mathbf{c}$.
\item $\mathbf{c}(s) + \sum\limits_{\rho \in R} \mathbf{u}(\rho) \Delta\rho(s) \geq 0$ for every $s \in \Lambda$.
\end{enumerate}
If a flux vector $\mathbf{u}$ is applicable at state $\mathbf{c}$, we can \textit{apply $\mathbf{u}$ to $\mathbf{c}$}, resulting in the state
\begin{center}
$\mathbf{c} * \mathbf{u} = \mathbf{c} + \sum\limits_{\rho \in R} \mathbf{u}(\rho) \Delta\rho$.
\end{center}
Equivalently, $\mathbf{c} * \mathbf{u} = \mathbf{c} + \mathbf{M}\mathbf{u}$. A \textit{flux vector sequence}, $\mathbf{U} = (\mathbf{u}_1,...,\mathbf{u}_k)$ is a tuple of flux vectors. We apply a flux vector sequence $\mathbf{U} = (\mathbf{u}_1,...,\mathbf{u}_k)$ iteratively to a state $\mathbf{c}$,
\begin{equation*}
\mathbf{c} * \mathbf{U} = (\mathbf{c} * (\mathbf{u}_1,\ldots,\mathbf{u}_{k-1})) * \mathbf{u}_k.
\end{equation*}  
A flux vector sequence $\mathbf{U} = (\mathbf{u}_1,...,\mathbf{u}_k)$ is \textit{applicable} at state $\mathbf{c}$ if $\mathbf{u}_i$ is applicable at $(\mathbf{c} * (\mathbf{u}_1,\ldots,\mathbf{u}_{i-1}))$ for every $1 < i \leq k$. If $\mathbf{c}$ and $\mathbf{d}$ are any states, we say that \textit{$\mathbf{d}$ is reachable from $\mathbf{c}$}, denoted $\mathbf{c} \rightarrow^* \mathbf{d}$, if there exists a flux vector sequence $\mathbf{U}$ applicable at $\mathbf{c}$ such that $\mathbf{c} * \mathbf{U} = \mathbf{d}$. We say that $\mathbf{d}$ is reachable from $\mathbf{c}$ \textit{in k steps}, denoted $\mathbf{c} \rightarrow^{k} \mathbf{d}$, if there exists a flux vector sequence $\mathbf{U} = (u_1,\ldots, u_k)$ applicable at $\mathbf{c}$ such that $\mathbf{c} * \mathbf{U} = \mathbf{d}$. A reaction $\rho \in R$ is \textit{eventually applicable from $\mathbf{c}$} if there exists a state $\mathbf{d}$ reachable from $\mathbf{c}$ so that $\rho$ is applicable at $\mathbf{d}$. A reaction is \textit{permanently inapplicable from $\mathbf{c}$} if it is not eventually applicable from $\mathbf{c}$. 

The following theorem, proven in \cite{CDS}, will be used in the proof of our first main theorem.
\begin{reftheorem}
If $\mathbf{c} \rightarrow^* \mathbf{d}$, then $\mathbf{c} \rightarrow^{m + 1} \mathbf{d}$ where $m = \vert R \vert$ is the number of reactions.
\end{reftheorem}

\section{The Reachability Problem for Continuous CRNs}
Having defined the relevant concepts for continuous chemical reaction networks, we are now able to formally define our problem CCRN-REACH.
\begin{ctsReachdefinition}
Given a continuous CRN $C = (\Lambda, R)$ and two states $\mathbf{c}$, $\mathbf{d} \in \mathbb{Q}^{\Lambda}$, output a flux vector sequence $\mathbf{U}$ such that $\mathbf{U}$ is applicable at $\mathbf{c}$ and $\mathbf{c} * \mathbf{U} = \mathbf{d}$, if one exists; output ``not reachable" otherwise.
\end{ctsReachdefinition}

Note that this problem becomes a trivial solution of a system of linear equations if we drop the requirement that the the flux vector sequence must be applicable at $\mathbf{c}$. We will prove that CCRN-REACH is computable in polynomial time. Intuitively, the dramatic difference in the computational difficulty between the VAS reachability problem (known to be at least EXPSPACE-hard) and CCRN-REACH is the additional flexibility given by the rational valued flux vectors. To compute CCRN-REACH, we show how to build a flux vector sequence lead from the starting state to a state of maximal support. This is only possible in the CCRN model of chemical reaction networks, which allows arbitrarily small additions via flux vectors. Once we are in such a maximal state we are able to get to the end state with the application of a single flux vector. To formalize this intuition, we will introduce several definitions and lemmas. 

Fix a continuous CRN $C = (\Lambda$, $R)$.  

\begin{definition}
Let $\mathbf{c}$ be a state, and $\epsilon > 0$. We say that a vector $\mathbf{u}$ is an \textbf{$\epsilon$-max support flux vector of $\mathbf{c}$} if $\mathbf{u}$ satisfies the following:
\begin{enumerate}
\item $\mathbf{u}$ is a flux vector that is applicable at $\mathbf{c}$.
\item for every flux vector $\mathbf{v}$ applicable at $\mathbf{c}$, $supp(\mathbf{c} * \mathbf{v}) \subseteq supp(\mathbf{c} * \mathbf{u})$.
\item $\| \mathbf{u} \| \leq \epsilon$.
\end{enumerate}
\end{definition}
That is, a vector is an $\epsilon$-max support flux vector of a state $\mathbf{c}$ if it is applicable at $\mathbf{c}$ and maximally increases the support of $\mathbf{c}$ while giving at most $\epsilon$ flux to each reaction. We will show that $\epsilon$-max support flux vectors exist for every state and $\epsilon > 0$. 

Let $\epsilon > 0$. We now construct a specific $\epsilon$-max support flux vector of $\mathbf{c}$, which we will henceforth call \textbf{the} $\epsilon$-max support flux vector of $\mathbf{c}$. Define $App_{\mathbf{c}}$ to be the set of all applicable reactions at $\mathbf{c}$. Let $\epsilon_{\mathbf{c}} = min\{\mathbf{c}(s)\, \vert \, s \in supp(\mathbf{c})\}$ (the lowest nonzero concentration of any species at state $\mathbf{c}$), $\Gamma_{\mathbf{c}} = max\{1$, $\vert \Delta\rho(s) \vert : \rho \in App_{\mathbf{c}}\}$, and $\delta_{\mathbf{c}, \epsilon} = \frac{1}{\Gamma_c \vert R \vert} min\{\frac{\epsilon_\mathbf{c}}{2},  \epsilon\}$. 

\begin{definition}
\textbf{The $\epsilon$-max support flux vector of $\mathbf{c}$} is the vector $\mathbf{u}_{\mathbf{c}, \epsilon}$ defined by

\begin{equation*}
  \mathbf{u}_{\mathbf{c}, \epsilon}(\rho)=\begin{cases}
    \delta_{\mathbf{c}, \epsilon}, & \text{if $\rho \in App_{\mathbf{c}}$}\\
    0, & \text{otherwise}
  \end{cases}
\end{equation*}
for every $\rho \in R$. 
\end{definition}

The following lemma shows that $\mathbf{u}_{\mathbf{c}, \epsilon}$ is a well defined $\epsilon$-max support flux vector of $\mathbf{c}$.

\begin{lemma}
Let $\mathbf{c}$ be a state, and $\epsilon > 0$. Then $\mathbf{u}_{\mathbf{c}, \epsilon}$ is an $\epsilon$-max support flux vector of $\mathbf{c}$.
\end{lemma}

When the context is clear, we will refer to $\mathbf{u}_{\mathbf{c}, \epsilon}$ as the max support flux vector. The following observation can be easily seen from the definition of the max support flux vector.
\begin{observation}
The $\epsilon$-max support flux vector of $\mathbf{c}$,  $\mathbf{u}_{\mathbf{c}, \epsilon}$, is computable in polynomial time in terms of $(C, \mathbf{c}, \epsilon)$.
\end{observation}

Since $\mathbf{u}_{\mathbf{c}, \epsilon}$ is a flux vector applicable at $\mathbf{c}$, we are able to discuss the max support flux vector of the state $(\mathbf{c} * \mathbf{u}_{\mathbf{c}, \epsilon})$. For convenience, we will use the following notation:
\begin{enumerate}
\item $\mathbf{u}_{\mathbf{c}, \epsilon}^1 := \mathbf{u}_{\mathbf{c}, \epsilon}$.
\item $\mathbf{u}_{\mathbf{c}, \epsilon}^k := $ the $\epsilon$-max support flux vector of the state $(\mathbf{c} * (\mathbf{u}_{\mathbf{c}, \epsilon}^1,...,\mathbf{u}_{\mathbf{c}, \epsilon}^{k-1}))$.
\end{enumerate} 
It is important to note that the vectors $\mathbf{u}^i_{\mathbf{c}, \epsilon}$ are distinct; in fact, the hope is that the set of applicable reactions grows with successive applications max support flux vectors.

\begin{definition}
Let $\epsilon > 0$, $m = \vert R \vert + 1$ and $\gamma = \frac{\epsilon}{m}$. The \textbf{$\epsilon$-max support flux vector sequence of $\mathbf{c}$}, denoted $\mathbf{U}_{\mathbf{c}, \epsilon}$, is defined to be the sequence
\begin{center}
$\mathbf{U}_{\mathbf{c}, \epsilon} = (\mathbf{u}_{\mathbf{c}, \gamma}^1,\ldots, \mathbf{u}_{\mathbf{c}, \gamma}^{m})$.
\end{center}
\end{definition}

From Observation 1 it is clear that $\mathbf{U}_{\mathbf{c}, \epsilon}$ is computable in polynomial time in terms of $(C, \mathbf{c}, \epsilon)$. 

\begin{observation}
For any state $\mathbf{c}$ and any $\epsilon > 0$, the $\epsilon$-max support flux vector sequence of $\mathbf{c}$ is a flux vector sequence that is applicable at $\mathbf{c}$. Moreover, $\|\sum\limits_{i = 1}^m\mathbf{u}_{\mathbf{c}, \gamma}^i\| \leq \epsilon$.
\end{observation}
\begin{proof}
This follows immediately from Lemma 1 and the choice of $\gamma$.
\end{proof}

The choice of restricting the length of the flux vector $\mathbf{U}_{\mathbf{c}, \epsilon}$ to $\vert R \vert + 1$ is not arbitrary. We will show that this is all that is required to get to the largest possible support of a state. 

\begin{definition}
Let $\mathbf{c}$ be a state and $\epsilon > 0$. We say that a state $\mathbf{m}$ is an \textbf{$\epsilon$-max support state of $\mathbf{c}$} if, for every state $\mathbf{d}$ that is reachable from $\mathbf{c}$, $supp(\mathbf{d}) \subseteq supp(\mathbf{m})$.
\end{definition}

Similar to our previous definitions, we now define \textbf{the} $\epsilon$-max support state of $\mathbf{c}$ to be
\begin{center}
$\mathbf{m}_{\mathbf{c}, \epsilon} := (\mathbf{c} * \mathbf{U}_{\mathbf{c}, \epsilon})$.
\end{center}

\begin{lemma}
If $\mathbf{c}$ is a state and $\epsilon > 0$, then $\mathbf{m}_{\mathbf{c}, \epsilon}$ is an $\epsilon$-max support state of $\mathbf{c}$.
\end{lemma}

Since $\epsilon$ was arbitrary in Lemma 2, we see that for every $\epsilon, \epsilon^\prime > 0$, $supp(\mathbf{m}_{\mathbf{c}, \epsilon}) = supp(\mathbf{m}_{\mathbf{c}, \epsilon^\prime})$. Recall that a reaction $\rho$ is eventually applicable from a state $\mathbf{c}$ if $\rho$ is applicable at some state $\mathbf{d}$ that is reachable from $\mathbf{c}$. By Lemma 2, a reaction $\rho$ is eventually applicable from a state $\mathbf{c}$ if and only if $\rho$ is applicable at $\mathbf{m}_{\mathbf{c}, \epsilon}$ for any $\epsilon > 0$. This allows us to compute all the permanently inapplicable reactions from $\mathbf{c}$, which will be vital in the algorithm computing CCRN-REACH.

\begin{observation}
The set of all permanently inapplicable reactions from $\mathbf{c}$ is computable in polynomial time.
\end{observation}
\begin{proof}
By Observation 1 we compute the $1$-max support state of $\mathbf{c}$, $\mathbf{m}_{\mathbf{c}, 1}$, and eliminate all reactions not applicable at $\mathbf{m}_{\mathbf{c}, 1}$.
\end{proof}

We are now ready to prove our first main theorem.
\begin{maintheorem}
CCRN-REACH is computable in polynomial time.
\end{maintheorem}
\begin{proof}
Consider the following algorithm (Algorithm 1 below) deciding CCRN-REACH.
\begin{algorithm}
\caption{CCRN-REACH on input $C = (\Lambda, R)$, $\mathbf{c}$, $\mathbf{d}$}
\label{CHalgorithm}
\begin{algorithmic}[1]
\State Eliminate from $R$ all permanently inapplicable reactions from $\mathbf{c}$
\For{each reaction $\rho \in R$}
\State Compute a vector $F_{\rho} \in \mathbb{Q}^R_{\geq 0}$ such that $\mathbf{c} + \mathbf{M} F_\rho = \mathbf{d}$ and $F_{\rho}(\rho) > 0$, if one exists
\State if no such vector exists, eliminate $\rho$ from R, GOTO 1.
\EndFor
\State if R = $\emptyset$, output ``not reachable"
\State otherwise define vector $S \in \mathbb{Q}^R_{\geq 0}$ as follows
\State for each $\rho \in R$, set $S(\rho) = \frac{1}{\vert R \vert}\sum\limits_{i=1}^{\vert R \vert} F_i(\rho)$
\State Compute $\epsilon = \frac{min\{S(\rho)\}_{\rho \in R}}{2}$
\State Compute the max support flux vector sequence $\mathbf{U}_{\mathbf{c}, \epsilon}$
\State Compute $\mathbf{v} = S - \mathbf{U}_{\mathbf{c}, \epsilon}$
\State Output $(\mathbf{U}_{\mathbf{c}, \epsilon}, \mathbf{v})$ (padded with $0$s for eliminated reactions)
\end{algorithmic}
\end{algorithm}
From our previous observations, it is clear that the algorithm runs in polynomial time in terms of the input. We now prove that $\mathbf{d}$ is reachable from $\mathbf{c}$ if and only if the above algorithm outputs a flux vector sequence $\mathbf{U}$ applicable at $\mathbf{c}$ such that $\mathbf{c} * \mathbf{U} = \mathbf{d}$. 

Assume that, on input $C = (\Lambda, R)$, $\mathbf{c}$ and $\mathbf{d}$, the algorithm outputs a sequence of vectors $\mathbf{U}$. Let $R$ be the set of reactions left after exiting the loop (necessarily non-empty), and $m = \vert R \vert$. By the choice of $\epsilon$ and Observation 2, for each $\rho \in R$,
\begin{equation*}
\sum\limits_{i = 1}^{m+1} \mathbf{u}_{\mathbf{c}, \gamma}^i(\rho) < S(\rho),
\end{equation*}
where $\mathbf{U}_{\mathbf{c}, \epsilon} = (\mathbf{u}_{\mathbf{c}, \gamma}^1,\ldots, \mathbf{u}_{\mathbf{c}, \gamma}^{m+1})$ (recall that $\gamma = \frac{\epsilon}{m + 1}$). Therefore, the vector $\mathbf{v} = S - \mathbf{U}_{\mathbf{c}, \epsilon}$ is a flux vector (in fact $\mathbf{v}$ is strictly positive). Hence the output $\mathbf{U} = (\mathbf{U}_{\mathbf{c}, \epsilon}, \mathbf{v})$ is a flux vector sequence. By Observation 2, $\mathbf{U}_{\mathbf{c}, \epsilon}$ is applicable at $\mathbf{c}$. Upon exiting the loop we are guaranteed that any reactions remaining in $R$ must be eventually applicable from $\mathbf{c}$ using only the other remaining reactions. Let $\rho \in supp(\mathbf{v})$. Then $\rho \in R$, and so $\rho$ must be eventually applicable from $\mathbf{c}$ using only reactions remaining in $R$. By Lemma 2, $(\mathbf{c} * \mathbf{U}_{\mathbf{c}, \epsilon}) = \mathbf{m}_{\mathbf{c}, \epsilon}$ is a max support state, therefore $\rho$ is applicable at $(\mathbf{c} * \mathbf{U}_{\mathbf{c}, \epsilon})$. Since $\rho$ was arbitrary, $\mathbf{v}$ is applicable at $(\mathbf{c} *\mathbf{U}_{\mathbf{c}, \epsilon})$, and so $(\mathbf{U}_{\mathbf{c}, \epsilon}, \mathbf{v})$ is a flux vector sequence that is applicable at $\mathbf{c}$. Finally, we have 
\begin{align*}
\mathbf{c} * (\mathbf{U}_{\mathbf{c}, \epsilon}, \mathbf{v}) & = \mathbf{c} + \mathbf{M}(\mathbf{U}_{\mathbf{c}, \epsilon} + \mathbf{v})\\
& = \mathbf{c} + \mathbf{M}S \\
& = \mathbf{c} + \mathbf{M}\frac{1}{\vert R \vert}\sum\limits_{i=1}^{\vert R \vert} F_i(\rho) \\
& = \mathbf{c} + \frac{1}{\vert R \vert}\sum\limits_{i=1}^{\vert R \vert} \mathbf{M}F_i(\rho) \\
& = \mathbf{c} + \frac{1}{\vert R \vert}\sum\limits_{i=1}^{\vert R \vert} \mathbf{d} - \mathbf{c} \\
& = \mathbf{d},
\end{align*}
where $\mathbf{M}$ is the stoichiometry matrix of $C = (\Lambda, R)$. Therefore if the algorithm outputs a vector sequence, then $\mathbf{d}$ is reachable from $\mathbf{c}$.

For the other direction, assume that $\mathbf{d}$ is reachable from $\mathbf{c}$. Then, by definition, there is a nonempty subset $R^\prime \subseteq R$ such that, for all $\rho \in R^\prime$,
\begin{enumerate}
\item $\rho$ is eventually applicable from $\mathbf{c}$ using only reactions from $R^\prime$, and
\item there exists a vector $F_\rho$ such that $\mathbf{M}F_\rho = \mathbf{d} - \mathbf{c}$ and $F_\rho(\rho) > 0$.
\end{enumerate}
Therefore the algorithm will exit the loop with $R$ nonempty, and output a flux vector sequence $(\mathbf{U}_{\mathbf{c}, \epsilon}, \mathbf{v})$. As we just shown, $(\mathbf{U}_{\mathbf{c}, \epsilon}, \mathbf{v})$ applicable at $\mathbf{c}$ such that $\mathbf{c} * (\mathbf{U}_{\mathbf{c}, \epsilon}, \mathbf{v}) = \mathbf{d}$. 

\end{proof}

\section{The Subset Reachability Problem}
Define the decision problem Sub-CCRN-REACH as follows, 
\begin{subsetReachdefinition}
Given a continuous CRN $C = (\Lambda, R)$, states $\mathbf{c}$, $\mathbf{d}$ and integer $k$, accept if and only if there exists a path from $\mathbf{c}$ to $\mathbf{d}$ using only $k$ reactions from $R$.
\end{subsetReachdefinition}

In contrast to the computational ease of CCRN-REACH, we give evidence that the related problem Sub-CCRN-REACH is quite difficult.
\begin{maintheorem}
Sub-CCRN-REACH is NP-complete.
\end{maintheorem}
\begin{proof}
From the proof that CCRN-REACH is computable in polynomial time, it is easy to see that Sub-CCRN-REACH is in NP. Simply guess a subset of $k$ reactions and decide CCRN-REACH on the subset. We will reduce 3SAT to Sub-CCRN-REACH to show hardness.

Let $\phi$ be a boolean formula on n variables $x_1,\ldots, x_n$ with m clauses $C_1,\ldots,C_m$. Construct an equivalent CCRN $C_{\phi} = (\Lambda_{\phi}, R_{\phi})$ and states $\mathbf{c}_\phi$, $\mathbf{d}_\phi$ as follows. 

For each $x_i$, define three species $S_i$, $s_i$ and $\bar{s_i}$, and the following four reactions, where $\emptyset$ is a null species (the reactants are being consumed without generating any products).
\begin{multicols}{4}
\begin{enumerate}
\item $S_i \rightarrow s_i$
\item $S_i \rightarrow \bar{s_i}$
\item $s_i \rightarrow \emptyset$
\item $\bar{s_i} \rightarrow \emptyset,$
\end{enumerate}
\end{multicols}

For each clause $C_j$ define one species $T_j$, and the following (catalytic) reactions
\begin{multicols}{2}
\begin{enumerate}
\item $s_i \rightarrow s_i + T_j$, for every $x_i$ in $C_j$ 
\item $\bar{s_i} \rightarrow \bar{s_i} + T_j$, for every $\bar{x_i}$ in $C_j$
\end{enumerate}
\end{multicols}

Define the start state $\mathbf{c}_\phi$ to have a concentration of 1 for each species $S_i$ and a concentration of $0$ for all other species. Define the end state $\mathbf{d}_\phi$ to have a concentration of $1$ for each species $T_j$, and a concentration of $0$ for all other species. 

We now show that $\phi \in$ 3SAT if and only if $\mathbf{c}_\phi \rightarrow^* \mathbf{d}_\phi$ using exactly $2n + m$ reactions. Assume $\phi \in$ 3SAT, let $\mathbf{x}$ be any satisfying assignment. Define the flux vector $\mathbf{u}_1$ by,
\begin{center}
$\mathbf{u}_1(\rho) = \begin{cases} 
	1 &\mbox{if } \rho = S_i \rightarrow s_i$ and $\mathbf{x}(x_i) = 1 \\ 
	1 &\mbox{if } \rho = S_i \rightarrow \bar{s_i}$ and $\mathbf{x}(x_i) = 0 \\ 
	0 & \mbox{otherwise.} \end{cases}
$
\end{center}
The flux vector $\mathbf{u}_1$ transfers all of the concentration of $S_i$ into either $s_i$ or $\bar{s_i}$, depending on the satisfying assignment $\mathbf{x}$. The number of reactions given positive flux in $\mathbf{u}_1$ is $n$. For each clause $C_j$, choose one variable $x_{i_j}$ or its negation that evaluates to true under $\mathbf{x}$. Define the flux vector $\mathbf{u}_2$ by,
\begin{center}
$\mathbf{u}_2(\rho) = \begin{cases} 
	1 &\mbox{if } \rho = s_{i_j} \rightarrow s_{i_j} + T_j$ and $x_{i_j}$ is the chosen variable from $C_j \\ 
	1 &\mbox{if } \rho = \bar{s_{i_j}} \rightarrow \bar{s_{i_j}} + T_j$ and $\bar{x_{i_j}}$ is the chosen variable from $C_j \\ 
	0 & \mbox{otherwise.} \end{cases} 
$
\end{center}
Therefore $\mathbf{u}_2$ only gives positive flux to $m$ reactions, one for each clause in $\phi$. Finally, define the flux vector $\mathbf{u}_3$ by,
\begin{center}
$\mathbf{u}_3(\rho) = \begin{cases} 
	1 &\mbox{if } \rho = s_i \rightarrow \emptyset$ and $\mathbf{x}(x_i) = 1 \\ 
	1 &\mbox{if } \rho = \bar{s_i} \rightarrow \emptyset$ and $\mathbf{x}(x_i) = 0 \\ 
	0 & \mbox{otherwise.} \end{cases} 
$
\end{center}
The flux vector $\mathbf{u}_3$ eliminates the concentrations of each species $s_i$ or $\bar{s_i}$ (only one of which has concentration $1$). Clearly $\mathbf{u}_3$ gives positive flux to only $n$ reactions. Hence $\mathbf{U} = (\mathbf{u}_1, \mathbf{u}_2, \mathbf{u}_3)$ gives positive flux to only $2n + m$ distinct reactions, and $\mathbf{c}_\phi * \mathbf{U} = \mathbf{d}_\phi$.

Assume $\mathbf{c}_\phi * \mathbf{U} = \mathbf{d}_\phi$. Since $S_i$ has concentration 0 at $\mathbf{d}$ at least one of the reactions $S_i \rightarrow s_i$, $S_i \rightarrow \bar{s_i}$ must be used, that is, at least $n$ reactions. Similarly, since $s_i$ and $\bar{s_i}$ have concentration 0 at $\mathbf{d}_\phi$, at lest $n$ must be used in any flux vector sequence. Since $T_j$ has concentration 1 at $\mathbf{d}_\phi$ at least one reaction of the form $s_i \rightarrow s_i + T_j$ or $\bar{s_i} \rightarrow \bar{s_i} + T_j$ must be used for each $T_j$, so at least $m$ reactions. Hence $\mathbf{U}$ must give positive flux to at least $2n + m$ reactions. In order to reach $\mathbf{d}$ using the minimal number of reactions, $2n + m$, $\mathbf{U}$ must only give flux to one of $S_i \rightarrow s_i$ or $S_i \rightarrow \bar{s_i}$. Let $\mathbf{x}$ be the assignment of the variables $(x_1,\ldots, x_n)$ given by
\begin{center}
$\mathbf{x}(x_i) = \begin{cases} 
	1 &\mbox{if } \mathbf{U}(S_i \rightarrow s_i) = 1 \\ 
	0 & \mbox{otherwise.} \end{cases} 
$
\end{center}
Since $\mathbf{c}_\phi * \mathbf{U} = \mathbf{d}_\phi$, $\mathbf{U}$ gives positive flux to $s_i \rightarrow s_i + T_j$ or $\bar{s_i} \rightarrow \bar{s_i} + T_j$ for each species $T_j$ and some $i$. Therefore each clause $C_j$ must be satisfiable under assignment $\mathbf{x}$. Hence, if $\mathbf{c}_\phi \rightarrow^* \mathbf{d}_\phi$ using exactly $2n + m$ reactions, then $\phi \in$ 3SAT.
\end{proof}

\section*{Acknowledgments}
We thank Tim McNicholl, Xiang Huang, Titus Klinge, and Jim Lathrop for useful discussions.

\section{Appendix}

\begin{prooflemma1}
First we show $\mathbf{u}_{\mathbf{c}, \epsilon}$ is a flux vector that is applicable at $\mathbf{c}$. It is clear that $\mathbf{u}_{\mathbf{c}, \epsilon}$ is a flux vector. It suffices to show that every reaction $\rho \in supp(\mathbf{u}_{\mathbf{c}, \epsilon})$ is applicable at $\mathbf{c}$ and that $(\mathbf{c} * \mathbf{u}_{\mathbf{c}, \epsilon}) \in \mathbb{R}_{\geq 0}^{\Lambda}$. From the definition of $\mathbf{u}_{\mathbf{c}, \epsilon}$, $\mathbf{u}_{\mathbf{c}, \epsilon}(\rho) > 0$ if and only if $\rho \in App_{\mathbf{c}}$. Therefore if $\rho \in supp(\mathbf{u}_{\mathbf{c}, \epsilon})$, then $\rho$ is applicable at $\mathbf{c}$. To complete the proof of item (1) we show $\mathbf{c} * \mathbf{u}_{\mathbf{c}, \epsilon}$ remains non-negative. Let $s \in \Lambda$ be any species, and assume that the concentration of $s$ at $\mathbf{c}$ is greater than $0$, i.e., $s \in supp(\mathbf{c})$. By the definition of $\mathbf{u}_{\mathbf{c}, \epsilon}$,
\begin{center}
$\vert\sum\limits_{\rho \in App_{\mathbf{c}}}{\mathbf{u}_{\mathbf{c}, \epsilon}(\rho) \Delta\rho(s)} \vert \leq \delta_{\mathbf{c}, \epsilon}  \vert R \vert  \Gamma_{\mathbf{c}}  \leq \frac{\epsilon_{\mathbf{c}}}{2}$
\end{center} 
and therefore, 
\begin{align*}
(\mathbf{c} * \mathbf{u}_{\mathbf{c}, \epsilon})(s) &= \mathbf{c}(s) + \sum\limits_{\rho \in R} \mathbf{u}_{\mathbf{c}, \epsilon}(\rho) \Delta\rho(s) \\
&\geq \mathbf{c}(s) - \vert\sum\limits_{\rho \in App_{\mathbf{c}}}{\mathbf{u}_{\mathbf{c}, \epsilon}(\rho) \Delta\rho(s)} \vert \\
&\geq \epsilon_{\mathbf{c}} - \frac{\epsilon_{\mathbf{c}}}{2} \\ 
&> 0.
\end{align*}
Hence for every species $s \in supp(\mathbf{c})$, $(\mathbf{c} * \mathbf{u}_{\mathbf{c}, \epsilon})(s) > 0$. Now assume $s \notin supp(\mathbf{c})$; the concentration of $s$ at $\mathbf{c}$ is $0$. As we have seen, the only reactions $\rho$ such that $\mathbf{u}_{\mathbf{c}, \epsilon}(\rho) > 0$ are those reactions which are applicable at $\mathbf{c}$. By our assumption $\mathbf{c}(s) = 0$, any applicable reaction $\rho$ at $\mathbf{c}$ must have $\Delta\rho(s) \geq 0$. It is therefore clear that 
\begin{center}
$(\mathbf{c} * \mathbf{u}_{\mathbf{c}, \epsilon})(s) \geq 0$.
\end{center}

We now prove that, for every flux vector $\mathbf{v}$ applicable at $\mathbf{c}$, $supp(\mathbf{c} * \mathbf{v}) \subseteq supp(\mathbf{c} * \mathbf{u})$. Let $s \in supp(\mathbf{c} * \mathbf{v})$. We first assume that $s \in supp(\mathbf{c})$. We showed previously that if $s \in supp(\mathbf{c})$, then $(\mathbf{c} * \mathbf{u}_{\mathbf{c}, \epsilon})(s) > 0$. Hence $s \in supp(\mathbf{c} * \mathbf{u}_{\mathbf{c}, \epsilon})$. Now assume that $s \notin supp(\mathbf{c})$. As the concentration of $s$ at $\mathbf{c}$ is $0$, any applicable reaction $\rho$ at $\mathbf{c}$ must have $\Delta\rho(s) > 0$. Since $s \in supp(\mathbf{c} * \mathbf{v})$, there must be at least one reaction $\rho_{s}$ applicable at $\mathbf{c}$ such that $\Delta\rho_{s}(s) > 0$. Since $\rho_s$ is applicable at $\mathbf{c}$, $\mathbf{u}_{\mathbf{c}, \epsilon}(\rho_s) = \delta_{\mathbf{c}, \epsilon}$. Thus
\begin{align*}
(\mathbf{c} * \mathbf{u}_{\mathbf{c}, \epsilon})(s) &= \mathbf{c}(s) + \sum\limits_{\rho \in R} \mathbf{u}_{\mathbf{c}, \epsilon}(\rho) \Delta\rho(s) \\
&= \vert\sum\limits_{\rho \in App_{\mathbf{c}}}{\mathbf{u}_{\mathbf{c}, \epsilon}(\rho) \Delta\rho(s)} \vert \\
&\geq \mathbf{u}_{\mathbf{c}, \epsilon}(\rho_s) \Delta\rho_s(s) \\
&\geq \delta_{\mathbf{c}, \epsilon} \\ 
&> 0.
\end{align*}

Finally, it is immediate that $\| \mathbf{u}_{\mathbf{c}, \epsilon} \| \leq \epsilon$.
\end{prooflemma1}

\begin{prooflemma1}
Let $\mathbf{d}$ be a state reachable from $\mathbf{c}$. By Theorem 0, there exists a flux vector sequence of length $m = \vert R \vert + 1$ taking $\mathbf{c}$ to $\mathbf{d}$, i.e. $\mathbf{c} \rightarrow^m \mathbf{d}$. By induction and use of Lemma 1, we see that for every state $\mathbf{d}$ such that $\mathbf{c} \rightarrow^m \mathbf{d}$, $supp(\mathbf{d}) \subseteq supp(\mathbf{m}_{\mathbf{c}, \epsilon})$. 
\end{prooflemma1}

\begin{thebibliography}{9}

\bibitem{AADFP}
	Dana Angluin, James Aspnes, Zoe Diamadi, Michael J. Fischer, and Rene Peralta. 
	Computation in networks of passively mobile finite-state sensors, 
	\emph{PODC ’04: Proceedings of the Twenty Third Annual ACM Symposium on Principles of Distributed Computing}, pages 290–299. ACM Press, 2004.
	
\bibitem{Cardelli}
	Luca Cardelli,
	Strand algebras for DNA computing,
	\emph{Natural Computing}, 10(1):407-428, 2011.
	
\bibitem{CDS}
	Ho-Lin Chen, David Doty, and David Soloveichik,
	Rate-independent computation in continuous chemical reaction networks,
	\emph{ITCS 2014: Proceedings of the 5th Innovations in Theoretical Computer Science Conference}, pp 313-326.
	
\bibitem{CDSPC}
	Y.-J. Chen, N. Dalchau, N. Srinivas, A. Phillips, L. Cardelli, D. Soloveichik, and G. Seelig,
	Programmable chemical controllers made from DNA, 
	\emph{Nature nanotechnology}, 
	8(10):755–762, 2013.

\bibitem{CSWB}
	Matthew Cook, David Soloveichik, Erik Winfree, and Jehoshua Bruck,
	Programmability of chemical reaction networks, 
	\emph{Algorithmic Bioprocesses}, 
	Springer Berlin Heidelberg,
	pages 543–584,
	2009.
	
\bibitem{EN}
	Javier Esparza and Mogens Nielsen, 
	Decidability issues for Petri Nets – a survey,
	\emph{Journal of Information Processes and Cybernetics}, 3:143–160, 1994.
	
\bibitem{ELSS}
	Michael B. Elowitz, Arnold J. Levine, Eric D. Siggia, and Peter S. Swain,
	Stochastic gene expression in a single cell, 
	\emph{Science}, 297:1183–1185, 2002.
\bibitem{gillespie}
	Daniel T. Gillespie, 
	Exact stochastic simulation of coupled chemical reactions, 
	\emph{Journal of Physical Chemistry}, 81(25):2340–2361, 1977
	
\bibitem{JRP}
	Hua Jiang, Marc Riedel, and Keshab Parhi, 
	Digital signal processing with molecular reactions,
	\emph{IEEE Design and Test of Computers}, 29(3):21–31, 2012

\bibitem{KM}
	Richard M. Karp and Raymond E. Miller, 
	Parallel program schemata, 
	\emph{Journal of Computer and System Sciences}, 3(4):147–195, 1969.
	
\bibitem{kosaraju}
	S. Rao Kosaraju, 
	Decidability of reachability in vector addition systems (preliminary version), 
	\emph{STOC 1982}, pages 267–281. ACM, 1982.

\bibitem{lambert}
	Jean Luc Lambert, 
	A structure to decide reachability in petri nets, 
	\emph{Theoretical Computer Science},99(1):79–104, 1992.

\bibitem{leroux}
	J\'er\^ome Leroux,
	Vector addition reachability problem (a simpler solution),
	\emph{The Alan Turing Centenary Conference},
	volume 10 of EPiC Series, pages 214–228. EasyChair, 2012.
	
\bibitem{lipton}
  Richard J. Lipton,
  The reachability problem requires exponential space,
  Technical report
  1976.
  
\bibitem{mayr}
	Ernst W. Mayr, 
	An algorithm for the general petri net reachability problem, 
	\emph{STOC 1981},
	pages 238–246. ACM, 1981.

\bibitem{MA}
	Harley H. McAdams and Adam P. Arkin, 
	Stochastic mechanisms in gene expression,
	\emph{Proceedings of the National Academy of Sciences}, 94:814–819, 1997.


\bibitem{ST}
	George S. Sacerdote and Richard L. Tenney,
	The decidability of the reachability problem for vector addition systems (preliminary version), 
	\emph{STOC 1977}, pages 61–76. ACM, 1977.


\bibitem{SCWB}
	David Soloveichik, Matthew Cook, Erik Winfree, and Jehoshua Bruck, 
	Computation with finite stochastic chemical reaction networks 
	\emph{Natural Computing}, 7(4):615–633, 2008.


\bibitem{SSW}
	David Soloveichik, Georg Seelig, and Erik Winfree, 
	DNA as a universal substrate for chemical kinetics, 
	\emph{Proceedings of the National Academy of Sciences}, 
	107(12):5393, 2010.

\end{thebibliography}
\end{document}